\documentclass[10pt,twocolumn,letterpaper]{article}
\usepackage[accsupp]{axessibility} 

\usepackage[pagenumbers]{cvpr} 

\usepackage{graphicx}
\usepackage{amsmath}
\usepackage{amssymb}
\usepackage{amsthm}
\usepackage{booktabs}
\usepackage{pifont}
\usepackage{algorithm}
\usepackage{algpseudocode}
\usepackage{mathtools}
\usepackage{paralist}
\usepackage{multirow}

\newtheorem{theorem}{Theorem}
\newtheorem{corollary}{Corollary}[theorem]
\newtheorem{lemma}[theorem]{Lemma}

\usepackage[pagebackref,breaklinks,colorlinks]{hyperref}
\newcommand{\cmark}{\ding{51}}%
\newcommand{\xmark}{\ding{55}}%

\usepackage[capitalize]{cleveref}
\crefname{section}{Sec.}{Secs.}
\Crefname{section}{Section}{Sections}
\Crefname{table}{Table}{Tables}
\crefname{table}{Tab.}{Tabs.}

\def\cvprPaperID{10320} 
\def\confName{CVPR}
\def\confYear{2022}

\begin{document}

\title{PILC: Practical Image Lossless Compression with \\
an End-to-end GPU Oriented Neural Framework }

\author{Ning Kang $^{1}$\thanks{ Equal contribution. Work done when Shanzhao Qiu interns at Huawei.} \quad Shanzhao Qiu $^{2}$\footnotemark[1]  \quad Shifeng Zhang $^{1}$ \quad Zhenguo Li $^{1}$\thanks{ Correspondence to: Zhenguo Li (li.zhenguo@huawei.com) and Shutao Xia (xiast@sz.tsinghua.edu.cn).} \quad Shutao Xia $^{2}$\footnotemark[2]\\
$^{1}$ Huawei Noah’s Ark Lab \quad $^{2}$ Tsinghua University \\
{\tt\small kang.ning2@huawei.com, qiusz20@mails.tsinghua.edu.cn}}


\maketitle

\begin{abstract}
Generative model based image lossless compression algorithms have seen a great success in improving compression ratio.
However, the throughput for most of them is less than 1 MB/s even with the most advanced AI accelerated chips, preventing them from most real-world applications, which often require 100 MB/s.
In this paper, we propose PILC, an end-to-end image lossless compression framework that achieves 200 MB/s for both compression and decompression with a single NVIDIA Tesla V100 GPU, $10\times$ faster than the most efficient one before. 
To obtain this result, we first develop an AI codec that combines auto-regressive model and VQ-VAE which performs well in lightweight setting, then we design a low complexity entropy coder that works well with our codec.
Experiments show that our framework compresses better than PNG by a margin of 30\% in multiple datasets.
We believe this is an important step to bring AI compression forward to commercial use.
\end{abstract}

\section{Introduction}
\label{sec:intro}


Lossy compression have shown great success in recent research\cite{DBLP:conf/iccv/ChoiEL19, DBLP:journals/corr/abs-1908-04187,DBLP:journals/corr/abs-2002-04988,DBLP:conf/icml/BlauM19,DBLP:conf/icml/RippelB17,DBLP:journals/corr/abs-1907-08310,DBLP:conf/cvpr/MentzerATTG18,conf/wacv/PatelAM21}. In this paper, we focus on lossless compression. 
The basic idea of a lossless compression algorithm is to represent more likely appeared data with shorter codewords, while less frequent data with longer ones, such that the codeword is shorter than the original data in expectation. 
For example, an image with each pixel randomly generated is rarely seen in the real world, while a real picture taken by a camera occurs much more often, so the latter compresses better for almost all image compression algorithms. 

\begin{figure}
    \centering
    \includegraphics[width=0.466\textwidth]{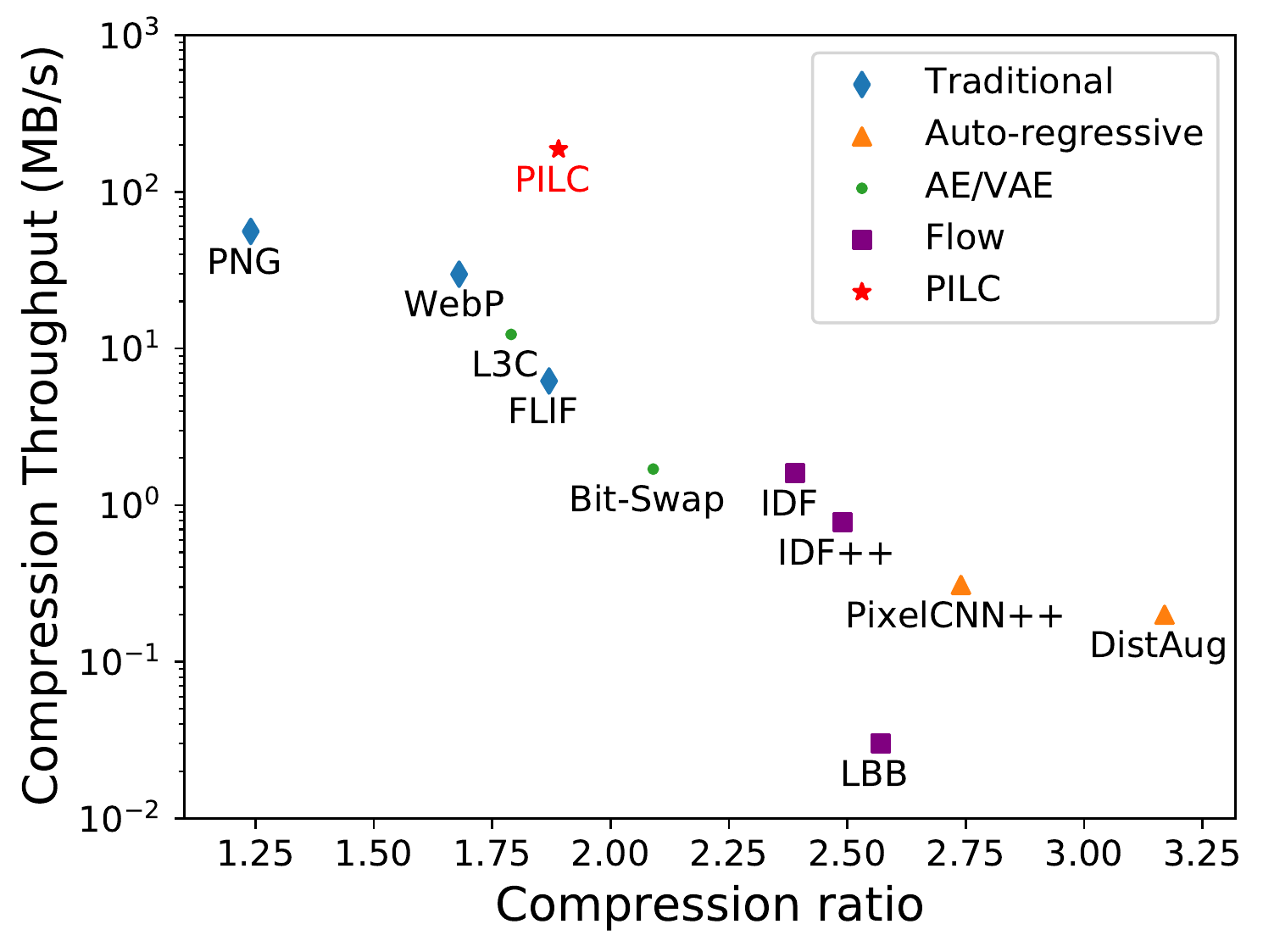}
    \caption{Comparison of compression throughput and ratio on CIFAR10. Our framework achieves high compression throughput with a competitive compression ratio.}
    \label{fig:enc_speed}
    \vspace{-4mm}
\end{figure}

According to Shannon's source coding theorem\cite{DBLP:journals/bstj/Shannon48}, no algorithm can compress data shorter than its entropy.
On the other hand, if the distribution of data is known in advance, an entropy coder can be applied so that the code length would be very close to its entropy, thus obtaining the optimal compression ratio.

Unfortunately, in most cases, the data distribution is unknown. 
Traditional algorithms make use of some prior information to infer this distribution. 
For example, LZ77~\cite{journals/tit/ZivL77} and LZ78~\cite{journals/tit/ZivL78} used by most archive formats assume that data with lots of repeated segments occur more often than those without. For image compression algorithms, PNG~\cite{journals/rfc/rfc2083} assumes most adjacent pixels are similar to each other, and JPEG-like algorithms apply the fact that lower frequency components are more significant than higher frequency ones.

Recently, deep generative models have shown great success in probabilistic modelling and lossless compression.
However, the low throughput makes them difficult to be used in real scenarios. 
An end user does not want to wait for too long to compress/decompress a file, and a file server needs to manipulate tons of files every day. In many cases, a 100+ MB/s throughput is needed. For example, streaming a 1080p video with 30 fps requires 187 MB/s for decompression. However, it is difficult for a learned algorithm to achieve this speed due to the following bottlenecks:

\textbf{Network inference.} Barely any AI compression algorithms have inference speed of 100+ MB/s because:
\begin{compactitem}
\item A large network is usually needed for a good density estimation~\cite{conf/iclr/Child21}.
\item Multiple network inferences are required to decompress with an auto-regressive~\cite{conf/icml/OordKK16, conf/iclr/SalimansK0K17} model.
\item Bits-Back scheme~\cite{journals/cj/FreyH97} prevents algorithms from applying a large batch size~\cite{conf/iclr/TownsendBB19}.
\end{compactitem}


\textbf{Coder.} Generative models only cooperate well with a dynamic entropy coder, which is far slower than static variants used in traditional algorithms. For most AI algorithms, the coder decodes only about 1 MB/s single-threaded.

\textbf{Data transfer.} Data transfer between CPU memory and GPU memory can also be a bottleneck for AI algorithms. 
\begin{compactitem}
\item Too many transfers are needed for decompression with an auto-regressive or hierarchical AE/VAE model. 
\item Transfer amount is too large for models predicting with a complicated distribution including lots of parameters, such as a distribution mixed by 10 logistic ones.
\end{compactitem}

Besides, other problems exist in current AI algorithms. One is single image decompression. One would expect to decompress only this image, instead of having to decompress many unrelated ones at the same time, but it is not the case for some VAE~\cite{conf/iclr/TownsendBB19} and Flow algorithms~\cite{conf/nips/HoLA19} using bits-back~\cite{journals/cj/FreyH97}. Another is compression with different image sizes. An algorithm that can only compress images with a fixed size hardly gets used in real applications, however, those with a fully connected layer or transformer layer in the network are not straight-forward to change the input size. \Cref{table:alg_limits} briefly shows the limitation of current algorithms.



\subsection{Our contributions}
In this paper, we focus on practical image lossless compression. To solve above issues, we make the following contributions:
\begin{compactitem}
\item We build an end-to-end framework with compression/decompression throughput of about 200 MB/s in one Tesla V100, and compression ratio 30\% better than PNG, which also enables single image decompression and image of different sizes.
\item We develop a very lite auto-regressive + Vector Quantized Variational Auto-encoder (VQ-VAE)~\cite{conf/nips/van17} model, whose inference speed is about 300 MB/s, but log-likelihood is similar to L3C~\cite{conf/cvpr/MentzerATTG19}, whose inference speed is only about 30 MB/s.
\item We design an AI compatible semi-dynamic entropy coder that is efficient in GPU besides CPU. As a result, only one data transfer with a minimal amount is needed between CPU memory and GPU memory. As far as we know, this is the first GPU coder applied in an AI compression algorithm.
\item We implement the end-to-end version in pure Python and PyTorch, making it easier for future works to test and improve on our implementation.
\end{compactitem}

\section{Related Work}

\begin{table}
\small
  \centering
    \caption{Limitations for some generative model based algorithms, where the columns mean: whether single image decompression is supported; whether images of different sizes are supported; whether inference time, coder, transfer time are faster than 100 MB/s, respectively.}
  \begin{tabular}{@{}lccccc@{}}
    \toprule
    Algorithm & Single & Size & Inference & Code & Trans \\
    \midrule
    PixelCNN~\cite{conf/icml/OordKK16} & \cmark & \cmark & \xmark & \xmark & \xmark \\
    DistAug~\cite{conf/icml/JunCCSRRS20} & \cmark & \xmark & \xmark & \xmark & \xmark \\
    L3C~\cite{conf/cvpr/MentzerATTG19} & \cmark & \cmark & \xmark & \xmark & \cmark \\
    BB-ANS~\cite{conf/iclr/TownsendBB19} & \xmark & \xmark & \xmark & \xmark & \xmark \\
    Bit-Swap~\cite{conf/icml/KingmaAH19} & \xmark & \xmark & \xmark & \xmark & \xmark \\
    IDF~\cite{conf/nips/HoogeboomPBW19} & \cmark & \cmark & \xmark & \xmark & \cmark \\
    LBB~\cite{conf/nips/HoLA19} & \xmark & \cmark & \xmark & \xmark & \xmark \\
    \textbf{PILC (Ours)} & \cmark & \cmark & \cmark & \cmark & \cmark \\
    \bottomrule
  \end{tabular}
  \label{table:alg_limits}
  \vspace{-4mm}
\end{table}


\textbf{Generative model based lossless compression.}
\label{subsec:AI lossless compression}
While human designed algorithms are sometimes very powerful, deep-learning models can usually do better (See \cref{fig:enc_speed}).
For example, to compress a portrait, many traditional algorithms can make use of the fact that the skin color is similar everywhere, but few know a human has two eyes, one nose, etc., which can be usually be captured by a generative model.
Denote the real probability of $x$ as $p(x)$, the predicted probability of $x$ as $q(x)$, and the expected codeword length as $L$. Then
\begin{equation}
    L(X) = -\sum_{x \in X} p(x) \log_2 q(x) = H(X) + KL(P || Q).\label{eq:actual_length}
\end{equation}
Therefore, for a deep-learning algorithm, the compression ratio depends on how accurate the prediction is. 
Generative model based image lossless compression algorithms can be classified into 3 types. 
\begin{compactitem}
\item Auto-regressive (AR) algorithms predict each symbol from the previous context.
Some of them are PixelCNN~\cite{conf/icml/OordKK16}, PixelCNN++~\cite{conf/iclr/SalimansK0K17}, DistAug~\cite{conf/icml/JunCCSRRS20}, etc..
They have great compression ratios with accurate modelling. 
However during decompression, since one pixel can be predicted only when all previous ones are known, one inference of the network can only decode one symbol, resulting in a very slow decompression speed.


\item Autoencoder (AE) and Variational autoencoder (VAE) algorithms make use of latent variables to help with the prediction. 
Since the whole image can be predicted with only one network inference, an auto-encoder algorithm such as L3C~\cite{conf/cvpr/MentzerATTG19} can decompress thousands of times faster than PixelCNN++~\cite{conf/iclr/SalimansK0K17}, but with a much worse compression ratio.
After bits-back~\cite{journals/cj/FreyH97} scheme is introduced to VAE algorithm~\cite{conf/iclr/TownsendBB19}, the compression ratio is improved significantly.
Thus later works such as BB-ANS~\cite{conf/iclr/TownsendBB19}, Bit-Swap~\cite{conf/icml/KingmaAH19} and HiLLoC~\cite{conf/iclr/TownsendBKB20} all apply this scheme.
Algorithms such as NVAE~\cite{conf/nips/VahdatK20} and VDVAE~\cite{conf/iclr/Child21} can even get ratios comparable to auto-regressive ones.
However, bits-back introduces a dependency between different images, making it impossible to decompress a single one. 
Furthermore, the cost of random initial bits can only be leveraged by compressing/decompressing a large number of images sequentially, affecting the batch size and transfer speed greatly.

\item Flow algorithms apply learnable invertible transforms to make data easier to compress. 
IDF~\cite{conf/nips/HoogeboomPBW19} and IDF++~\cite{conf/iclr/BergG0SS21} that have only integer transforms are named as integer flows. 
On the other hand, general flows such as LBB~\cite{conf/nips/HoLA19} and iFlow~\cite{journals/corr/abs-2111-00965} use transforms in real numbers.
There is also one type in between called volume preserving flow, which is discovered in iVPF~\cite{conf/cvpr/ZhangZ0L21}. Integer flows are the only type that do not require bits-back, hence the only one that supports single image decoding. But as the expressive power of a single layer is weaker, the model needs to be large, affecting the inference time. Besides, general flow models require one coding step for each layer, making data transfer another bottleneck.
\end{compactitem}


\begin{figure*}
  \centering
  \begin{subfigure}[t]{0.816\linewidth}
    \includegraphics[height=0.50\textwidth]{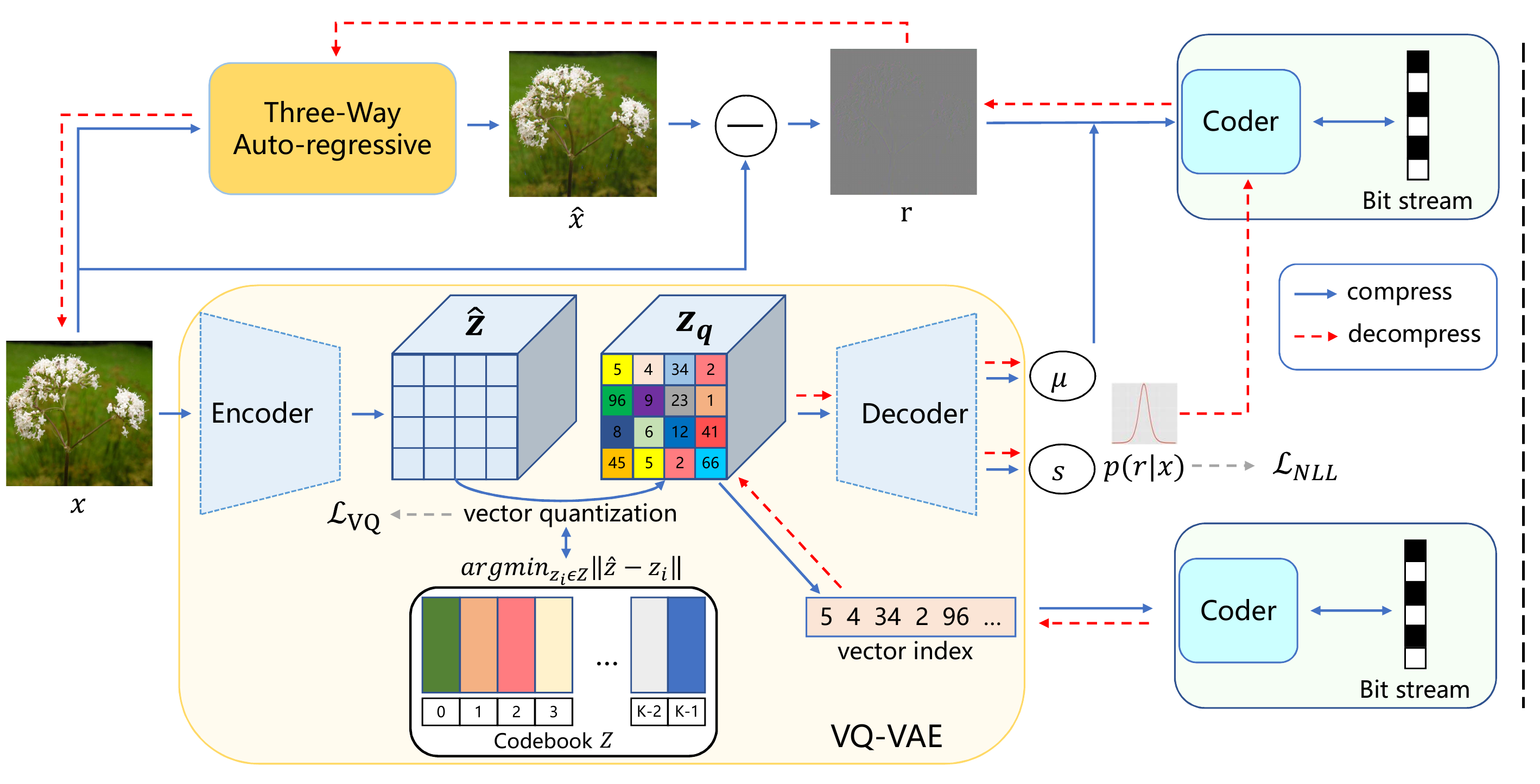}
    \caption{Overall framework of our proposed lossless algorithm}
    \label{fig:framework}
  \end{subfigure}
  \hfill
  \begin{subfigure}[t]{0.178\linewidth}
      \includegraphics[height=1.85\textwidth]{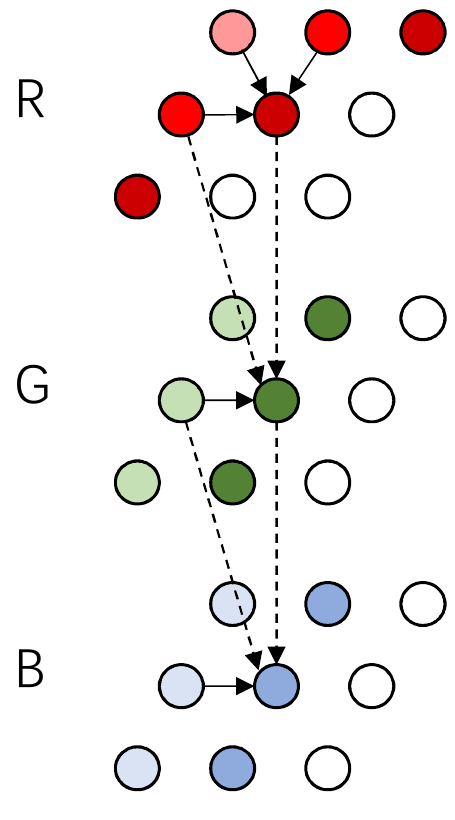}
    \caption{Three-Way Auto-regressive}
    \label{fig:three-way}
  \end{subfigure}
  \caption{Left: Our framework consists of Three-Way Auto-regressive, VQ-VAE, and Coder. The AR module predicts the input image and obtains the residual. The VQ-VAE models the probability distribution of the residual with a codebook. The Coder compresses image residual and vector index to bit stream. When decompress, vector index is decompressed first and pushed into the VQ-VAE decoder to obtain the residual distribution. This distribution is subsequently adopted to decompress the residual and then recover the original image (blue solid arrows indicate compress process, red dotted ones indicate decompress, best viewed in color). Right: Three-Way Auto-regressive adopts three predicted points to predict the current point.}
  \label{fig:overall}
  \vspace{-4 mm}
\end{figure*}

\textbf{Entropy Coder.}
A static entropy coder codes data from a fixed distribution. It is used by most commercial algorithms. 
A dynamic entropy coder, on the other hand, can code many distributions. It is usually used to code one kind of distribution with parameters as input, such as a normal distribution with variable means and variances.
The most widely used entropy coders such as Huffman Code~\cite{huffman}, Arithmetic Coder (AC)~\cite{arithmetic_ibm,arithmetic_stanford}, and Asymmetric Numerical System (ANS)~\cite{ans} all have both static and dynamic variants.
Dynamic entropy coders are much slower than static ones, as the density mass needs to be calculated on the fly, unlike the latter one, which can precalculate all of them beforehand.
For example, in a single-threaded setting, a static variant of ANS named FSE~\cite{fse} decodes faster than 300 MB/s, while another variant named rANS~\cite{ans} decodes at only about 1 MB/s with a distribution mixed by 10 logistic ones. 
Despite this fact, a dynamic coder is required for an AI algorithm to do fine-grained density estimation, which is also a key point that makes AI outperform traditional methods.


\section{Methods}
In this section, we introduce PILC, a practical image lossless compression framework. The proposed framework consists of a generative model and a semi-dynamic entropy coder. The model is composed of Three-Way Auto-regressive module and VQ-VAE module. We demonstrate how the model interact with the coder to achieve high throughput.

\subsection{Model Architecture}
\label{subsec:model_architecture}

\textbf{Overall Model.}
Among all current algorithms, AE have the highest throughput but the lowest compression ratio.
As suggested by many previous works, AE models cannot capture local features very well, while a local AR model is capable of doing so~\cite{conf/nips/SchirrmeisterZB20, journals/corr/abs-2109-02639}.
Therefore an intuitive idea would be to use a small AR to capture local features, and AE for more global ones. 
Since AR model hurts the parallelism and affect the throughput heavily, it must be very small.
In this work, a Three-Way Auto-regressive (TWAR) model is designed, where each pixel is predicted with only 3 other ones, reserving the ability to decompress in parallel. 
For the AE model, we replace it with VQ-VAE~\cite{conf/nips/van17}, which is proven to be successful in image generation~\cite{conf/nips/van17,vqvae2,vqgan,cogview}, while the inference time is still similar to AE.

As is shown in ~\cref{fig:framework}, TWAR predicts the original image and output image residual, then VQ-VAE estimates the distribution of the residual.
Distribution of values in the residual is much simpler than those of original images, making it easier to be predicted by VQ-VAE and coded with our entropy coder.



{\em Compress}: Given an image $x$, we use the Three-Way Auto-regressive module to predict the image, and get the reconstruction image $\hat{x}$, image residual $r = x - \hat{x}$. 
Then we adopt VQ-VAE to estimate the residual distribution $p(r|x)$, which we model as a simple logistic distribution. With $r$ and $p(r|x)$, we could compress $r$ to bit stream with our proposed dynamic entropy coder. We also compress the vector index to bit stream using the coder.

{\em Decompress}: At decompress time, we decompress vector index from bit stream. Then we select corresponding vectors from the codebook $Z$ to formulate the latent feature $z_q$. The feature is fed to the decoder and outputs the residual distribution $p(r|x)$. The coder subsequently decode the residual $r$ using $p(r|x)$. The AR module then decodes the original image $x$ from the residual $r$.

\textbf{Learning Residual with Three-Way Auto-regressive.} Inspired by the Paeth filter which is utilized in the traditional lossless compression method PNG\cite{journals/rfc/rfc2083} to transform the image to make it more efficiently compressible, we adopt a similar rule that utilizes three predicted points to predict the current one, which is shown in ~\cref{fig:three-way}. Our AR module is formulated by:
\begin{equation}
    \hat{x}_{ruv} = W_{r}^T \cdot (x_{r(u - 1)(v - 1)}, x_{r(u - 1)v}, x_{ru(v - 1)} ) + b_r, \label{eq:r-channel}
\end{equation}
\begin{equation}
    \hat{x}_{guv} = W_{g}^T \cdot(x_{gu(v - 1)}, x_{ru(v-1)}, x_{ruv} ) + b_g, \label{eq:g-channel}
\end{equation}
\begin{equation}
    \hat{x}_{buv} = W_{b}^T \cdot(x_{bu(v - 1)}, x_{gu(v-1)}, x_{guv} ) + b_b, \label{eq:b-channel}
\end{equation}
where $r$, $g$, $b$ denote image channel, $u$, $v$ denote the spatial location. We pad the input image with one row on the top and one column on the left to initiate the prediction.
For the red channel, the current point is predicted by its up, left, upper-left points like Paeth filter. Inspired by YCoCg-R color space which is shown successful in FLIF\cite{conf/icip/SneyersW16} format and utilizes different channels to calculate the color space, for green and blue channel, we also use their previous channel to conduct the prediction.
Different from \cite{effientAR}, our AR module models the pictures channel-wise, predict the mean only, and only has 12 parameters (with a bias in each channel), but it could still generate a high-quality prediction.
As shown in ~\cref{fig:framework}, $\hat{x}$ is quite close to $x$, which suggests the locality of natural images.

As explained in ~\cref{sec:intro}, auto-regressive models like PixelCNN~\cite{conf/icml/OordKK16} fail in decompression speed due to their dependencies nature because it is hard to parallelize the inference process. 
Fortunately, our AR module could be parallelized during inference.~\Cref{fig:ae_parallel} suggests one step in decoding. For the red channel, the current point is dependent on its up, left, upper-left points, so anti-diagonal points can be calculated in parallel. 
Once the red channel is recovered, for the green channel, we could calculate column by column to recover this channel, and the blue channel is the same. 
Thus, more points could be calculated at one time, which speeds up the inference process of the AR module. Because the padding rule is fixed, we could eventually losslessly recover the original image $x$ using the residual $r$. 

\textbf{Learning Distribution with VQ-VAE.} 
Vector quantization has been proved to be effective in lossy compression\cite{softtohard,learningcontent}.
Instead of designing a more sophisticated auto-encoder (AE) model, we utilize VQ-VAE\cite{conf/nips/van17} to learn residual distribution. Intuitively, given a fixed storage space, for VQ-VAE, we could just store the vector index and when decoding, with the codebook we could utilize more diversity features. But for vanilla AE, we only have the fixed size features. Experiments in \cref{sub:model_ablation} suggest VQ-VAE's superiority over vanilla AE. 
As the original image would contain much more spatial information, we choose to model the residual distribution given the original image, \ie, $p(r|x)$. 

As shown in \cref{fig:framework}, given an image $x \in \mathbb{R}^{H\times W\times 3}$, we get its latent feature by:
\begin{equation}
    \hat{z} = E(x) \in \mathbb{R}^{h\times w\times D},
\end{equation}
where $E$ is residual block based encoder, and $D$ is feature depth dimension. Then, with a codebook $Z = \{z_i\}_{i=1}^K \in \mathbb{R}^D $, we could conduct vector quantization on $\hat{z}$:
\begin{equation}
    z_q = q(\hat{z}) := \left(\mathop{\arg \min}\limits_{z_i \in Z}\parallel \hat{z}_{uv} - z_i \parallel \right) \in \mathbb{R}^{h\times w\times D},
    \label{eq:vector quantization}
\end{equation}
where $q$ is the element-wise quantization for each $\hat{z}_{uv} \in \mathbb{R}^D$ to its closest code $z_i$, and $u$, $v$ is the spatial location. Finally, we output the residual distribution $p(r|x)$ which we model as a logistic distribution and is formulated by:
\begin{equation}
    (\mu,  s) = D(z_q) \in \mathbb{R}^{H\times W\times 3 \times 2},
\end{equation}
where $\mu$ is the location parameter, $s$ is the scale parameter, and $D$ is residual block based decoder.

At compress time, the quantized vector index is encoded to bit stream by the coder. At decompress time, we could decode the index, select the corresponding vectors, and decode the distribution $p(r|x)$. In this way, we utilize the codebook $Z$ to memorize more information. 

Since there is no real gradient defined for \cref{eq:vector quantization}, the gradient back-propagation is achieved by straight-through estimator, which just copies the gradient from the decoder to encoder. The VQ loss is formulated by:
\begin{equation}
    \mathcal{L}_{VQ} = {\parallel \mathtt{sg}[\hat{z}] - z_q \parallel}^2 + \beta{\parallel \hat{z} - \mathtt{sg}[z_q] \parallel}^2,
    \label{eq: vq loss}
\end{equation}
where $\mathtt{sg}$ denotes the stop-gradient operator. The first term optimizes the codebook $Z$ and the second term is a commitment loss to make sure $\hat{z}$ commits to codebook and does not grow arbitrarily\cite{conf/nips/van17}. $\beta$ is a weight parameter to balance the two terms.

\begin{figure}
    \centering
    \includegraphics[width=0.40\textwidth]{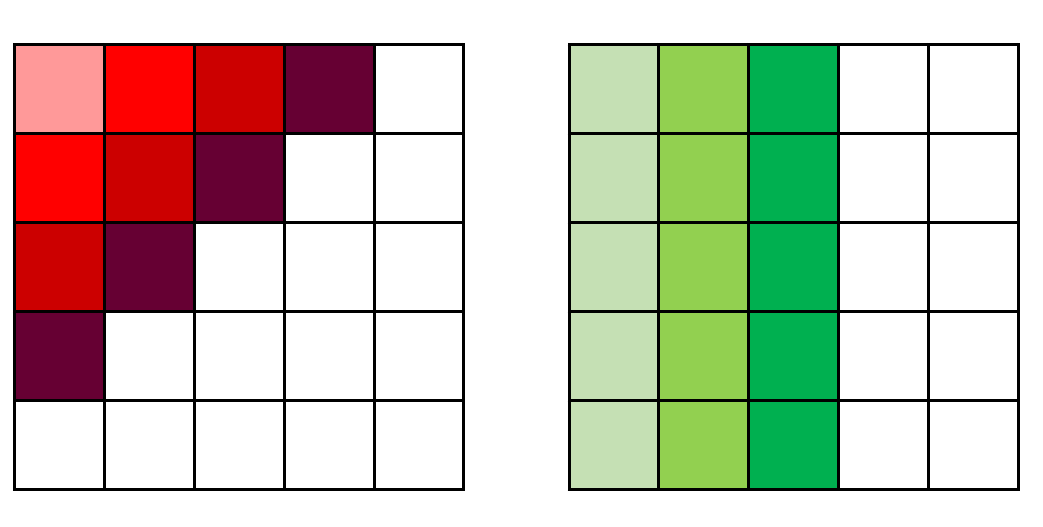}
    \caption{The decoding process of Red channel (left) and Green/Blue channel (right) for proposed AR model. For the red channel, anti-diagonal pixels can be decompressed in parallel, while for G and B, one column of pixels can be decoded in parallel.}
    \label{fig:ae_parallel}
    \vspace{-4mm}
\end{figure}

\textbf{Loss Function.} 
During training, we could jointly optimize the AR module and VQ-VAE module by maximizing the negative log-likelihood of the residual $r$:
\begin{equation}
     \mathcal{L}_{NLL} = \mathbb{E}_{x\sim p(x)} \left[ -\log p(r|x)\right],
     \label{eq:nll_loss}
\end{equation}
where $p(x)$ denotes the distribution of the training dataset. The lower the negative log-likelihood is, the fewer bits we need to compress the residual $r$. See ~\cref{subsec:inter_model_coder} for the more implement detail. 

Thus, the total objective function to optimize our neural compression model can be formulated by:
\begin{equation}
    \mathcal{L}_{total} = \mathcal{L}_{NLL} + \alpha \mathcal{L}_{VQ},
    \label{eq:total_loss}
\end{equation}
where $\alpha$ balances the two terms. 

Our lightweight model with the Three-Way Auto-regressive and VQ-VAE achieves a high inference speed, which plays an important role in our overall framework.

\subsection{Coder Design}
We design our entropy coder based on rANS~\cite{ans}, and save computation time by precalculating intermediate results.
Previous works such as tANS~\cite{ans} already adopt this idea. However, tANS can only code a fixed distribution, which is inappropriate for AI algorithms.
On the other hand, our proposed coder is more similar to dynamic ones, as it supports distributions with variable parameters.
The only difference is that the parameters need to be quantized, and probability mass needs to be calculated for all of them. Thus we call it a semi-dynamic entropy coder.

To code with rANS, one first needs to choose a constant $M$, then for each symbol $x$, quantize its probability mass function (denoted as $P_x$), so that $\sum_x P_x = 2^M$. Let $C_x$ be the cumulative distribution function of $x$. It satisfies that $C_x = \sum_{i<x} P_x$. To encode pixels within range $[0,255]$, $M$ is usually chosen in $[10, 15]$.

In this work, we modify rANS slightly, and propose \cref{alg:ans_encode} and \cref{alg:ans_decode} as the logic behind.

\begin{algorithm}
\caption{Modified rANS: Encode $x$}\label{alg:ans_encode}
\begin{algorithmic}[1]
\Require $M$, $P_x, C_x$, State $S$, stream
\State $S \gets S + 2^M$\label{alg_step:encode_pre_add}
\While{$S \ge P_x * 2$}
    \State $\mathtt{push\_stream}(S\mod 2)$ \Comment{last bit of $S$ to stream}
    \State $S \gets \lfloor S \div 2 \rfloor$
\EndWhile
\State $S \gets S - P_x + C_x$\label{alg_step:encode_post_add}
\end{algorithmic}
\end{algorithm}
\vspace{-4mm}

\begin{algorithm}
\caption{Modified rANS: Decode one symbol}\label{alg:ans_decode}
\begin{algorithmic}[1]
\Require $M$, $P_x, C_x$ for all $x$, State $S$, stream
\State $x \gets \mathtt{binary\_search}(S)$: $ C_x \le S < C_x + P_x$
\State $S \gets S - C_x + P_x$\label{alg_step:decode_pre_add}
\While{$S < 2^M$}
    \State $S \gets S \times 2 + \mathtt{pop\_stream()}$ \Comment{last bit of stream}
\EndWhile
\State $S \gets S - 2^M$\label{alg_step:decode_post_add}
\State \Return $x$
\end{algorithmic}
\end{algorithm}
\vspace{-3mm}

Loops and binary search form the time-consuming steps for these algorithms, so tables are built to cache these results.

For encoding, one solution would be to create a table to store the number of bits to be pushed for each possible value of $S$ and $P_x$. Each value can be represented as a 16-bit unsigned integer, so the memory needed is $2^{2M+1}$ bytes.
In this work, we assume the number of distributions for different symbols is small. Denote the total number of distributions as $D$, and the distribution index $d \in [0, D)$ for each symbol can be calculated via the network, and let $X$ be the number of colors. Since $S$ is always in $[2^M, 2^{M+1})$ after Step \ref{alg_step:encode_pre_add}, so for a fixed $d$, the number of loops can differ only 1. So we build a table $\delta$ with dimension $D \times X$ such that
\begin{equation}
    \frac{S}{2^{\lfloor(\delta[d, x] + S) / 2^M \rfloor}} \in [P_x, P_x * 2).
\end{equation}
To further simplify the algorithm, Step \ref{alg_step:encode_post_add} is merged with Step \ref{alg_step:encode_pre_add} of next symbol in this work, and table $\phi$ is built such that $\phi[d, x] = 2^M - P_x + C_x$. $S$ needs to be initialized in $[2^M, 2^{M+1})$ then. The encoding process now is shown in \cref{alg:tans_encode}. Both $\delta$ and $\phi$ use unsigned 16-bit integers, so the memory needed for tables are $4 \times D \times X$ bytes. As the total memory needed for all $P_x$'s and $C_x$'s are also $4 \times D \times X$ bytes, which is not needed any more, \cref{alg:tans_encode} does not cost any additional memory.
\vspace{-2 mm}
\begin{algorithm}
\caption{ANS-AI: Encode $x$}\label{alg:tans_encode}
\begin{algorithmic}[1]
\Require $M$, $\delta$, $\phi$, $x$, $d$, State $S$, stream

\State $b = \lfloor (\delta[d, x] + S) / 2^M \rfloor$
\State $\mathtt{push\_stream}(S, b)$ \Comment{last $b$ bits of S to stream}
\State $S \gets \lfloor S / 2^b \rfloor + \phi[d, x]$

\end{algorithmic}
\end{algorithm}
\vspace{-3 mm}

In terms of decoding, for each distribution index $d$ and state $S$, we maintain a table $\Theta$ to store the symbol to be decoded, a table $B$ for number of bits to pop, and a table $N$ to merge Step \ref{alg_step:decode_pre_add} to \ref{alg_step:decode_post_add}. Since $\Theta$ and $B$ use 8-bit unsigned integers, and $N$ use a 16-bit unsigned one, the total memory consumed for these tables are $D \times 2^{M+2}$ bytes. One alternative is to replace $N$ and $B$ with $\phi$ and $\delta$, reducing the memory to $D \times 2^M + 4 \times D \times X$, but this will make the computation more complicated, and reduce the decoding speed by about 50\% from our experiments.
\vspace{-2 mm}
\begin{algorithm}
\caption{ANS-AI: Decode one symbol}\label{alg:tans_decode}
\begin{algorithmic}[1]
\Require $M$, $\Theta$, $B$, $N$, $d$, State $S$, stream
\State $x \gets \Theta[d, S]$
\State $S \gets N[d,S] + \mathtt{pop\_stream}(B[d,S])$ \Comment{last $B[d,S]$ bits}
\State \Return $x$
\end{algorithmic}
\end{algorithm}
\vspace{-3 mm}

Our proposed coder not only simplifies the computation, but also makes it possible to implement with a machine learning framework such as PyTorch by removing all the logical units. 

\subsection{Integration of Model and Coder}
\label{subsec:inter_model_coder}
\begin{figure}[t]
    \centering
    \includegraphics[width=0.366\textwidth]{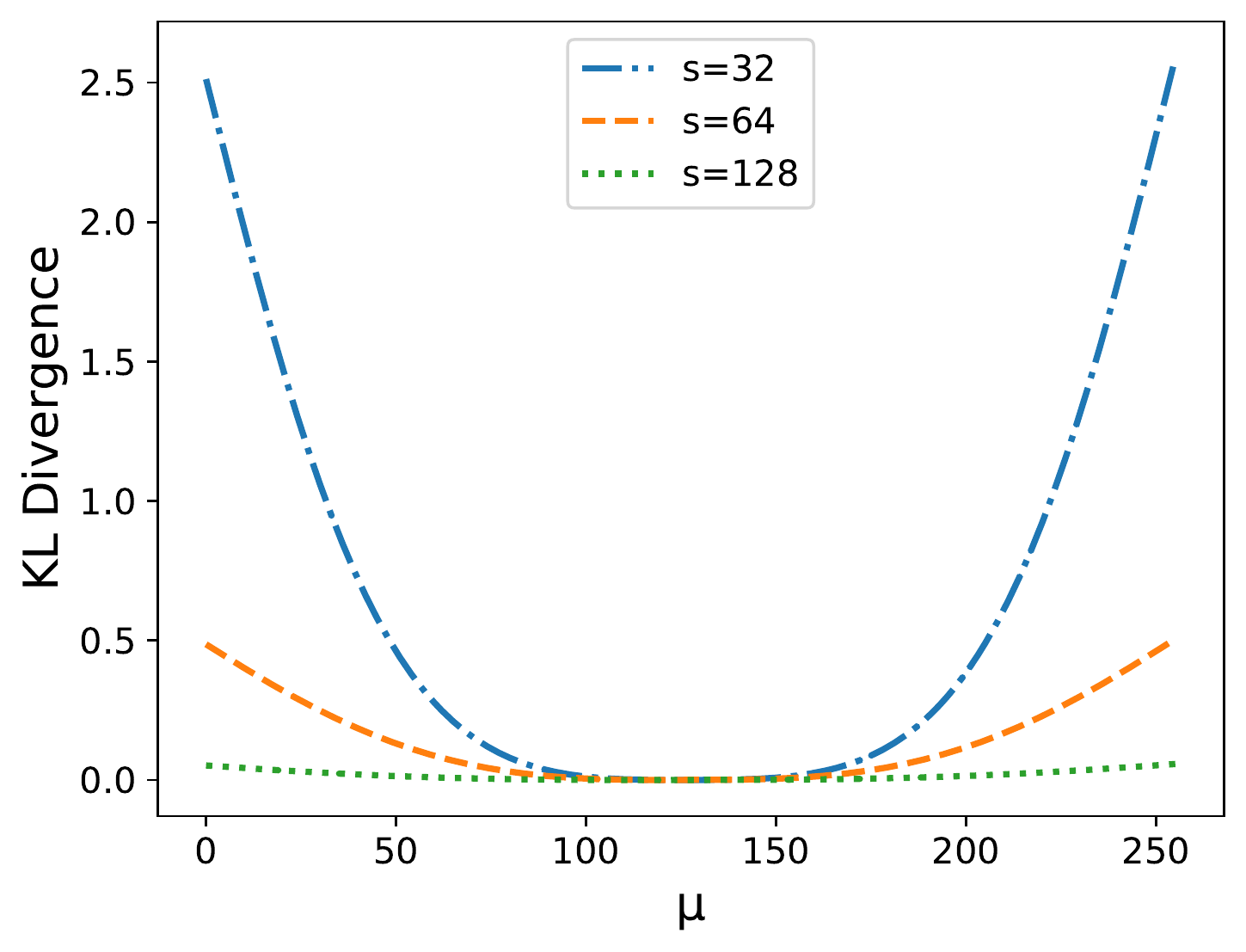}
    \caption{KL divergence between real truncated logistic distribution and approximated distribution, which indicates the BPD loss for the approximation. The figure shows that the approximated distribution is very close to the real one when $\mu \in [96,160]$.}
    \label{fig:kl_u}
    \vspace{-4 mm}
\end{figure}

To guarantee the effectiveness of the coder, the total number of distributions $D$ needs to be small.
In our model, each pixel of $r$ is predicted with a truncated logistic distribution with location parameter $\mu$ and scale parameter $s$, truncated between 0 and 255. Denote the distribution as $r \sim L(\mu, s, 0, 255)$. Since the number of means is required to be quantized to at least 256 values, $D$ would be over one thousand.


In this work, we predict $r - \mu + 128$ within \textbf{uint8} space, \ie, the prediction of $r \sim L(\mu, s, 0, 255)$ is approximated as $(r - \mu + 128) \bmod 256 \sim L(128, s, 0, 255)$. To make it as accurate as possible, $\mu$ needs to be as close to $128$ as possible. We apply the following two techniques for this:
\begin{compactenum}
    \item In VQ-VAE module, instead of predicting residual $r$ directly, we predict $(r + 128)\bmod 256$.
    \item We add a sigmoid layer to predict the final value of $\mu$ in VQ-VAE module to make sure it is close to 128.
\end{compactenum}

\cref{fig:kl_u} and \cref{fig:predicted_u} show that this technique almost deals with no BPD loss.

\section{Experiments}
In this section, we evaluate the model and coder effectiveness with extensive experiments and we conduct experiments on both low resolution and high resolution images which demonstrate our framework obtains the fastest compression speed with a competitive compression ratio. 

\subsection{Experimental Settings}
\label{subsec:exp_settings}
\textbf{Implementation Details.} The AR and VQ-VAE module are jointly trained. The AR module contains 12 parameters. In VQ-VAE, the codebook size $K$ is set to 256 and the codebook dimension $D$ is set to 32. 
The encoder and decoder both have 4 residual blocks, all with 32 channels. The encoder downsampling ratio is 2, and pixel shuffle is used when upsampling in the decoder. Thus, an image $x $ with shape $32 \times 32 \times 3$ would be quantized to $16 \times 16 \times 1$ vectors. 
$\beta$ in ~\cref{eq: vq loss} is set to 0.25 as in \cite{conf/nips/van17}, and $\alpha$ in ~\cref{eq:total_loss} is set to 125. 

We do all experiments in a Linux docker with Intel(R) Xeon(R) Gold 6151 CPU (3.0 GHz, 72 threads), and one NVIDIA Tesla V100 32GB.
For speed tests, we ignore the time for disk I/O, since it has nothing to do with compression algorithms.
But for fair comparison between AI and traditional methods, we do count the time for data transfer between CPU memory and GPU memory.
That means for all AI algorithms, timer starts when data is loaded into RAM and ends when compressed data is stored in RAM. 

Different from previous AI compression works which focus on density estimation, this paper targets practical compression. Therefore, for BPD values reported in this paper, all storage required to recover the original image is counted, including meta data, Huffman tables, auxiliary bits, etc.. For AI algorithms, real BPD is reported instead of theoretical one unless specified otherwise, which means for each image, we calculate the number of bytes (instead of bits) needed for each image, then take the average.

\begin{table}[t]
  \centering
    \caption{Ablation study on the model architecture. Theoretical BPD on CIFAR10 is reported.}
  \begin{tabular}{@{}lc@{}}
    \toprule
    Model & BPD \\
    \midrule
    Vanilla Auto-encoder (AE) & 5.53 \\
    VQ-VAE & 5.13\\ 
    Three-Way Auto-regressive (AR) & 4.92 \\
    AR + AE & 4.35 \\
    \textbf{AR + VQ-VAE} & \textbf{4.17}\\
    \bottomrule
  \end{tabular}
  \label{tab:arch_components}
  \vspace{-4 mm}
\end{table}

\begin{table}
\centering
\caption{Ablation study on the receptive field of AR model and the parallel mechanism. Theoretical BPD and decompress throughput on CIFAR10 (red channel) are reported. Receptive field means the current point is predicted by how many previous points.}
\label{tab:ar_receptive_field}
\begin{tabular}{ccr} 
\toprule
\multicolumn{1}{c}{Receptive Field} & BPD    & \multicolumn{1}{c}{\begin{tabular}[c]{@{}c@{}}Throughput\\(MB/s)\end{tabular}}  \\ 
\midrule
3 (with parallel)   & 5.7714 & \textbf{382.5}\\
3                   & 5.7714 & 48.5 \\
4                                  & 5.7737 & 47.5                                                                            \\
5                                  & 5.7708 & 46.0                                                                            \\
6                                  & 5.7701 & 44.8                                                                            \\
7                                  & 5.7449 & 44.0                                                                            \\
\bottomrule
\end{tabular}
\vspace{-4 mm}
\end{table}

\textbf{Datasets.}
For low resolution experiments, we train and evaluate on CIFAR10, ImageNet32/64. For high resolution experiments, we train our model on 50000 images from Open Image Train Dataset\cite{OpenImages} which are randomly selected and preprocessed as \cite{conf/cvpr/MentzerATTG19} to prevent overfitting on JPEG artifacts. Then, we evaluate on natural high resolution datasets, CLIC.mobile, CLIC.pro\cite{clic} and DIV2K\cite{div2k} to test the model generalization ability. For 
high resolution setting, we crop 32 $\times$ 32 patches for evaluating. 

\begin{figure}[t]
    \centering
    \includegraphics[width=0.366\textwidth]{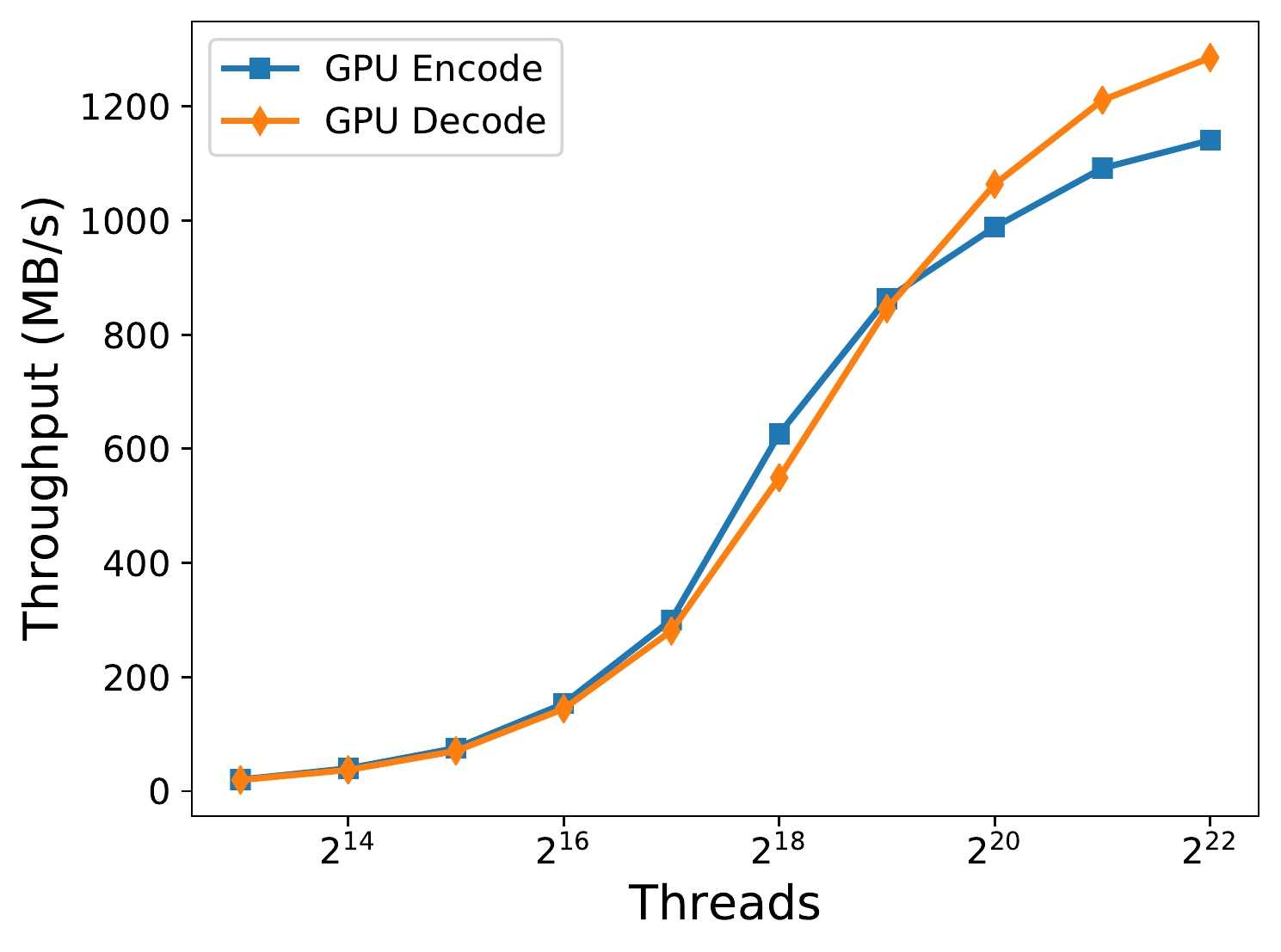}
    \caption{Encoding/decoding speed of the GPU coder. When the number of threads (batch size) is small, the throughput of the GPU coder is low, but it can benefit from large number of threads.}
    \label{fig:coder_speed}
    \vspace{-4 mm}
\end{figure}

\subsection{Model Effectiveness}
\label{sub:model_ablation}

\begin{table*}
\small
\centering
\caption{Compression performance in BPD and throughput in compression and decompression, compared to both engineered codec PNG\cite{journals/rfc/rfc2083}, WebP\cite{webp}, FLIF\cite{conf/icip/SneyersW16}, JPEG2000\cite{jpeg2000/books/daglib/0007442} and the learned L3C\cite{conf/cvpr/MentzerATTG19}.}
\begin{tabular}{lccc|ccc|rr} 
\toprule
\multicolumn{1}{c}{\multirow{2}{*}{BPD~}} & \multirow{2}{*}{CIFAR10} & \multirow{2}{*}{ImageNet32} & \multirow{2}{*}{ImageNet64} & \multirow{2}{*}{DIV2K} & \multirow{2}{*}{CLIC.pro} & \multirow{2}{*}{CLIC.mobile} & \multicolumn{2}{c}{Throughput (MB/s)~}  \\
\multicolumn{1}{c}{}  &    &     &     &     &     &    & Compress      & Decompress \\ 
\midrule
PNG\cite{journals/rfc/rfc2083} (fastest)  & 6.44    & 6.78    & 6.09   & 4.64    & 4.23     & 4.39   & 55.9     & 118.2  \\
PNG\cite{journals/rfc/rfc2083} (best)    & 5.91    & 6.41    & 5.77   & 4.23    & 3.90       & 3.80     &3.0     & 83.5 \\
WebP\cite{webp} (-z 0)  & 4.77    & 5.44    & 4.92      & 3.43    & 3.22   & 3.03   & 29.8     & 99.1  \\
FLIF\cite{conf/icip/SneyersW16} (--effort 0)   & 4.27  & \textbf{5.06} & 4.70 & 3.24    & 3.03       & 2.82    & 6.2      & 4.2   \\
JPEG2000\cite{jpeg2000/books/daglib/0007442}   & 6.75   &  7.50    &   6.08   & 4.11  &  3.79      &  3.94     &  7.6    & 9.1    \\
L3C\cite{conf/cvpr/MentzerATTG19}    &    4.55      & 5.19   &    \textbf{4.57}    &   \textbf{3.13}      &       \textbf{2.96}     &       \textbf{2.65}      &  12.3     &  6.3      \\
\textbf{PILC (Ours)}  & \textbf{4.23}      & 5.10    & 4.76     & 3.41   & 3.23   & 3.00    & \textbf{180.3} & \textbf{217.2} \\
\bottomrule
\label{tab:overall_speed}
\end{tabular}
\vspace{-4 mm}
\end{table*}

\textbf{Effectiveness of model architecture.} We analyze single and combinations of each component to evaluate the effectiveness of model architecture. In \cref{tab:arch_components}, the vanilla Auto-encoder (AE) contains the same encoder and decoder architecture as VQ-VAE with latent feature shape 16$\times$16$\times$1. In AE and VQ-VAE, we directly use the model to predict the distribution of the original image. In the single AR, we compress the residual with a static entropy coder to get the BPD. AR+AE follows the same process as AR+VQ-VAE. As shown in \cref{tab:arch_components}, VQ-VAE outperforms AE which demonstrates the effectiveness of the learned codebook $Z$. We could intuitively explain that 16$\times$16$\times$1 vector indexes would contain more diversity features than 16$\times$16$\times$1 numbers in AE's latent. Our proposed AR+VQ-VAE model achieves the best theoretical BPD on CIFAR10 which suggests the effectiveness of predicting residual.

\textbf{Effectiveness of Three-Way AR and parallel mechanism.} We compare different receptive fields of the AR model and the parallel mechanism. Experiments are conducted on red channel of CIFAR10 images. The receptive field of AR model means the current point is predicted by how many previous points. In our model, the receptive field is three. As shown in \cref{tab:ar_receptive_field}, as the receptive field increases, the BPD indeed descends, but the decline is very limited. With the parallel mechanism we designed in \cref{subsec:model_architecture}, the Three-Way AR achieves a throughput of 382.5 MB/s when recovering original images from residual in decompression, and outperforms other settings which cannot be paralleled.

\subsection{Coder Effectiveness}

We implement the GPU coder based on \cref{alg:tans_encode} and \cref{alg:tans_decode}. We test the coder on the randomly generated data (1000000 pieces, each with length 1000, and range 0-255). The signals have the same dimension with each value in 0-7. For each setting, we run 10 times and take the average speed. From \cref{fig:coder_speed}, we see that GPU coder is slow for the small number of threads. Fortunately, the number of threads can easily reach one million with a V100 32G. Therefore, we can inference on million batches, and save the result in GPU. Then we can code all of the batches at once.
This is beneficial for the GPU coder, in the sense that the data transfer times and the amount is smaller for both compression and decompression.

\subsection{Distribution Effectiveness}
\begin{figure}[t]
    \centering
    \includegraphics[width=0.38\textwidth]{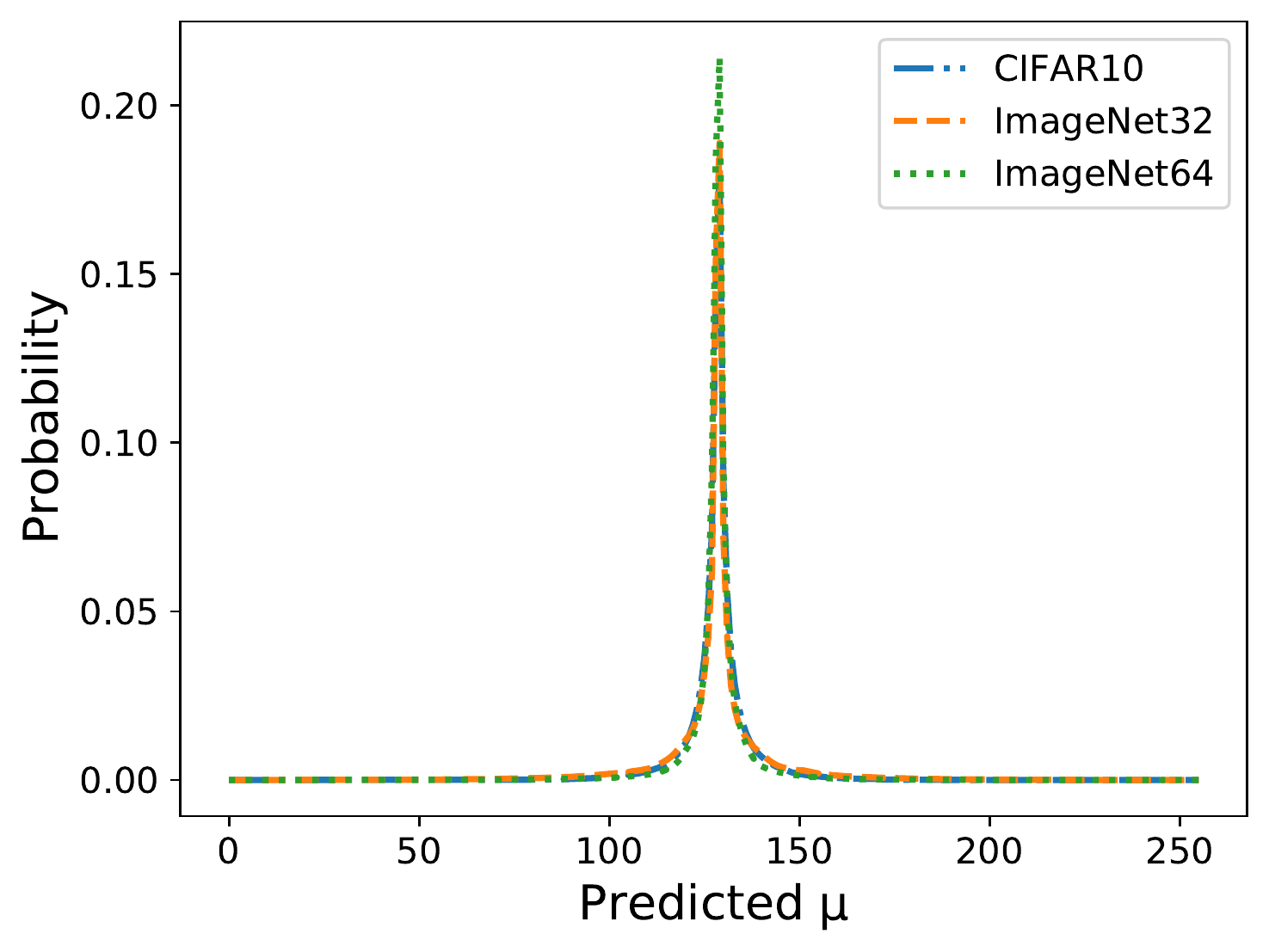}
    \caption{The probability mass of the predicted $\mu$ on multiple datasets. They roughly follow the double exponential distribution with location parameter 128. }
    \label{fig:predicted_u}
    \vspace{-4 mm}
\end{figure}
We count the predicted location parameter $\mu$ for the distribution specified in \cref{subsec:inter_model_coder}.
It can be shown in \cref{fig:predicted_u} that for all the three datasets, the predicted values nearly always fall in $[96, 160]$. As stated in \cref{subsec:inter_model_coder}, this cause almost no loss in BPD.

\subsection{Compression performance}

\textbf{Compare with exists methods.} We reproduce all methods in ~\cref{tab:overall_speed} on the same platform explained in \cref{subsec:exp_settings} for fair comparison. We use Python CV2 lib for PNG and JPEG2000. We adopt FLIF version 1.0.1 and use command `flif -e --effort=0' for fastest setting. For WebP, we use version 1.2.1 and the command `cwebp -lossless -z 0'. For L3C, we reproduce the results using the official model\cite{l3c_github} and we replace the original coder in L3C with ours to test the speed as the original coder is very slow. The throughput results are reported based on high resolution settings. ~\Cref{tab:overall_speed} shows our framework achieves the fastest compress and decompress throughput. Also, the framework achieves a competitive compression ratio on low resolution images and we still outperform PNG, WebP and JPEG2000 on high resolution images.

 
\textbf{End-to-end throughput.} To fully use the GPU coder and make the measured speed more accurate, we duplicate the CIFAR10 validation set 100 times (1 million images). ~\Cref{tab:decomposition} shows the details of the throughput and time taken for each step.

\begin{table}[t]
\small
\centering
\caption{Throughput \& decomposition of time for the whole compression/decompression process on CIFAR10. Throughput is calculated with respect to the size of the original data. Time is measured as $\mu$s per image.}
\begin{tabular}{ccrr} 
\toprule
\multicolumn{1}{l}{}        & Phase                 & \multicolumn{1}{c}{\begin{tabular}[c]{@{}c@{}}Throughput\\~ ~ (MB/s)\end{tabular}} & \multicolumn{1}{c}{\begin{tabular}[c]{@{}c@{}}Time\\($\mu$s)\end{tabular}}  \\ 
\midrule
\multirow{5}{*}{Compress}& RAM $\rightarrow$~GPU & 9246&0.33  \\
& Model Inference  & 276  & 11.11    \\
& Coder Encode  & 675    & 4.55       \\
& GPU $\rightarrow$~RAM & 2985   & 1.03 \\
& \textbf{Total}   & \textbf{180}   & \textbf{17.02}      \\ 
\midrule
\multirow{7}{*}{Decompress} & RAM $\rightarrow$~GPU & 11101  & 0.28  \\
 & Coder Decode & 11091   & 0.28   \\
 & VQ-VAE Decode  & 721   & 4.26    \\
& Code Decode     & 672    & 4.57    \\
& AR Decode  & 869     & 3.53      \\
 & GPU $\rightarrow$~RAM & 2521  & 1.20   \\
& \textbf{Total}    & \textbf{217}   & \textbf{14.12}     \\
\bottomrule
\end{tabular}
\label{tab:decomposition}
\vspace{-4 mm}
\end{table}

\subsection{Limitation Analysis}
Our framework, although very light-weighted, requires an AI chip such as Tesla V100 to work efficiently, which is inaccessible for most end-users. For example, one image with 2K resolution compresses only about 1.3 MB/s in the CPU of our server machine. Therefore, our framework is only applicable to file/cloud servers.


\section{Conclusion}
In this work, we develop the first efficient end-to-end generative model based image lossless compression framework with throughput of about 200 MB/s, and compression ratio 30\% better than PNG. To achieve this result, we develop a model that combines AR model and VQ-VAE, and a semi-dynamic coder that is computation efficient.

To further improve PILC, future works could focus on algorithms that are more friendly in everyday devices, such as a PC, or mobile phone. This requires a network with much less FLOPS than the current VQ-VAE one.
{\small
\bibliographystyle{ieee_fullname}
\bibliography{egbib}
}

\appendix

\section{Coder Design}
With the same notation as ours, the original rANS can be written as:
\begin{equation}\label{eq:original_rans_encoding}
    \text{Encoding: } S \gets 2^M \times \lfloor S / P_x \rfloor + C_x + S \bmod P_x,
\end{equation}
\begin{equation}\label{eq:original_rans_decoding}
    \text{Decoding: } S \gets P_x \times \lfloor S / 2^M \rfloor + S \bmod 2^M - C_x,
\end{equation}
where in decoding, same as the modified one, to get $x$, binary search is also needed to be applied on $S \bmod 2^M$.

Let $n$ be the message length, if applying \cref{eq:original_rans_encoding} and \cref{eq:original_rans_decoding} directly, the time complexity is $O(n^2)$. So in real applications, $S$ is usually truncated to a range, such as $[2^{32}, 2^{64})$, reducing the time complexity to $O(n)$ for encoding and $O(n \log X)$ for decoding, with minor reduction in compression ratio. However, it is still not efficient enough for real-time applications due to:
\begin{compactitem}
\item $P_x$ and $C_x$ need to be calculated online for each $x$.
\item One division operation and one modular operation are required for encoding.
\item Binary search is needed for decoding.
\end{compactitem}

In this work, $S$ is truncated to $[2^M, 2^{M+1})$, which boosts the efficiency, but affects the compression ratio more severely, and more memory is needed. In next subsections, we prove that the BPD loss and memory are both upper-bounded by an acceptable value.

\subsection{BPD loss for ANS-AI}
\begin{lemma}
BPD for the ANS-AI is at most $2 - \log_2 e$ worse than the original rANS.
\end{lemma}

\begin{proof}
As ANS-AI and the modified rANS are numerically equivalent, it suffices to prove the effectiveness of the latter.
Without loss of generality, assume quantization for pdf values does not introduce bias in density estimation, which means $p(x=i) = P_i / 2^M$ for all $i$, otherwise both original and modified ones are affected by the same way.
Also each symbol $x$ is independent of each other.
Unless specified otherwise, all $\log$ in this material means logarithm function with base 2. Then from Shannon's source coding theorem:
\begin{equation}\label{eq:bpd_ans}
\begin{split}
    \mathtt{BPD(rANS)} & \ge \frac{\sum_i P_i \times \log 2^M / P_i}{2^M} \\
    & = M - \frac{\sum_i P_i \times \log P_i}{2^M}.
\end{split}
\end{equation}

For the ANS-AI, BPD value equals to the expected number of bits to push to steam during encoding. The value $S$ before each encoding is uniform in $[2^M, 2^{M+1})$ since $P_i$ is proportional to $p(x=i)$. Therefore:
\begin{equation}\label{eq:bpd_ans_ai}
\begin{split}
    \mathtt{BPD(ANS\_AI)} & = \sum_{j \in [2^M, 2^{M+1})}\frac{\sum_i P_i \times \lfloor \log (j/P_i) \rfloor / 2^M} {2^M} \\
    & \le \frac{\sum_{j \in [2^M, 2^{M+1})} \sum_i P_i \times \log (j/P_i)}{2^{2M}} \\
    & = \frac{\sum_{j \in [2^M, 2^{M+1})}\log j}{2^M} - \frac{\sum_i P_i \times \log P_i}{2^M} \\
    & \le \frac{\int_{2^M}^{2^{M+1}} \log j}{2^M} - \frac{\sum_i P_i \times \log P_i}{2^M} \\
    & = M + 2 - \log e - \frac{\sum_i P_i \times \log P_i}{2^M}.
\end{split}
\end{equation}

Bring \cref{eq:bpd_ans} and \cref{eq:bpd_ans_ai} together, we get
\begin{equation}
    \mathtt{BPD(ANS\_AI) - BPD(rANS) \le 2 - \log e}.
\end{equation}
\end{proof}

\subsection{Memory consumption for ANS-AI}

\begin{lemma}\label{lemma:range_delta}
In our frame work, when $M \le 12$, $\delta$ can be represented as an 16 bit unsigned integer.
\end{lemma}

\begin{proof}
Let $k$ be the integer such that $P_x \times 2^k \in [2^M, 2^{M+1})$. In our framework, it always satisfies that $P_x \in [1, 2^{M-1})$, thus $k \in [2, M]$. Therefore:
\begin{equation}\label{eq:range_delta}
\lfloor \frac{\delta[d,x]+S}{2^M} \rfloor = \begin{cases}
k & S \in [P_x \times 2^k, 2^{M+1}) \\
k - 1 & S \in [2^M, P_x \times 2^k)
\end{cases}
\end{equation}

Then it suffices to prove that when $\delta[d,x] = k \times 2^M - P_x \times 2^k$, both unsigned 16 bit constraint and \cref{eq:range_delta} are satisfied.

\begin{enumerate}
    \item $\delta[d,x] < k \times 2^M - 2^M < (M - 1) \times 2^M \le 11 \times 2^{12} < 2^{16}$, and $\delta[d,x] \ge 2 \times 2^M - P_x \times 2^k > 0$. Therefore $\delta[d,x]$ is in range of unsigned 16 bit integer.
    \item If $S \in [P_x \times 2^k, 2^{M+1})$, then $\delta[d,x] + S \ge \delta[d,x] + P_x \times 2^k = k \times 2^M$, and $\delta[d,x] + S < k \times 2^M - P_x \times 2^k + 2^{M+1} \le k \times 2^M - 2^M + 2^{M+1} = (k+1)\times 2^M$. Therefore, in this case, $\lfloor \frac{\delta[d,x]+S}{2^M} \rfloor = k$.
    \item If $S \in [2^M, P_x \times 2^k)$, then $\delta[d,x] + S < \delta[d,x] + P_x \times 2^k = k \times 2^M$, and $\delta[d,x] + S >= k \times 2^M - P_x \times 2^k + 2^M > k \times 2^M - 2^{M+1} + 2^M = (k-1) \times 2^M$. Therefore, in this case, $\lfloor \frac{\delta[d,x]+S}{2^M} \rfloor = k - 1$.
\end{enumerate}
\end{proof}

\begin{corollary}\label{corollary:delta_general}
In general, when $M \le 11$, $\delta$ can be represented as a 16 bit signed integer.
\end{corollary}
If one is using distributions different from us, then the range of $P_x$ may  changed to $[1, 2^M)$, which means the probability of one symbol may be at least $0.5$, then the condition is changed to the one shown in Corollary \ref{corollary:delta_general}. The proof is almost the same with Lemma \ref{lemma:range_delta}.

\begin{corollary}
When $M \le 12$, ANS-AI encoding requires $4 \times D \times X$ bytes of memory, and decoding requires $D \times 2^{M+2}$ bytes.
\end{corollary}

\subsection{Speed comparison of rANS and ANS-AI}
We compare the speed of rANS coder and ours in the same machine as stated in experiment section of the main body. For fair comparison, both coders are implemented using C++ with OpenMP, and run on CPU. As is shown in \cref{table:throughput_rans_ai}, ANS-AI is 10 to 100 times faster than rANS.
\begin{table}[t]
\centering
\caption{Throughput comparison of rANS and ANS-AI. To generate the test data, first 8 logistic distributions are fixed. For each symbol, we first select a distribution uniformly randomly, and then sample data according to the distribution.}
\begin{tabular}{ccrr}
\toprule
& Threads & \multicolumn{1}{c}{\begin{tabular}[c]{@{}c@{}} Throughput\\ rANS (MB/s)\end{tabular}} & \multicolumn{1}{c}{\begin{tabular}[c]{@{}c@{}}Throughput \\ANS-AI (MB/s)\end{tabular}} \\
\midrule
\multirow{4}{*}{Encode} & 1 & 5.1 & 81.7 \\
 & 4 & 10.8 & 239.0 \\
 & 8 & 15.9 & 433.9 \\
 & 16 & 21.6 & 598.8 \\
\midrule
\multirow{4}{*}{Decode} & 1 & 0.8 & 122.0 \\
 & 4 & 2.8 & 467.9 \\
 & 8 & 5.5 & 925.9 \\
 & 16 & 7.4 & 1190.0 \\
\bottomrule
\end{tabular}
\label{table:throughput_rans_ai}
\end{table}

\section{GPU Coder vs CPU Coder}
We use the same algorithm we designed to create a CPU Coder. As shown in \cref{tab:cpu_coder_gpu_coder}, we compare the composition throughput between CPU Coder and GPU Coder. The CPU Coder is set to 16 threads. We duplicate the CIFAR10 validation set for ten times for experiment.

{\em Compress.} In terms of compression time, the throughput of CPU and GPU coders is comparable because both require two RAM-GPU transfers. We notice that GPU to RAM time is different between the two Coder. Because the residual data, latent data, and distribution parameters must all be transferred from GPU to RAM for CPU Coder, but only compressed data must be transferred from GPU to RAM for GPU Coder, which takes less time.

\begin{table*}[t]
\centering
\caption{Throughput \& decomposition of time for the whole compression/
decompression process on CIFAR10 of CPU Coder and GPU Coder. Throughput is calculated
with respect to the size of the original data. Time is measured
as $\mu$s per image.}
\label{tab:cpu_coder_gpu_coder}
\begin{tabular}{ccrr|crr} 
\toprule
& \multicolumn{3}{c|}{CPU Coder}  & \multicolumn{3}{c}{GPU Coder}   \\
& Phase & \multicolumn{1}{c}{\begin{tabular}[c]{@{}c@{}}Throughput\\(MB/s)\end{tabular}} & \multicolumn{1}{c|}{\begin{tabular}[c]{@{}c@{}}Time\\($\mu$s)\end{tabular}} & Phase           & \multicolumn{1}{c}{\begin{tabular}[c]{@{}c@{}}Throughput\\(MB/s)\end{tabular}} & \multicolumn{1}{c}{\begin{tabular}[c]{@{}c@{}}Time\\($\mu$s)\end{tabular}}  \\ 
\midrule
\multirow{5}{*}{Compress}   & RAM to GPU  & 9186   & 0.33   & RAM to GPU   & 9246  & 0.33 \\
 & Model Inference & 276  & 11.12 & Model Inference & 276  & 11.11 \\
& GPU to RAM  & 1158  & 2.65  & Coder Encode   & 675   & 4.55   \\
& Coder Encode  & 873  & 3.52  & GPU to RAM   & 2985 & 1.03  \\
& \textbf{Total} & \textbf{174}   & \textbf{17.62}   & \textbf{Total}  & \textbf{180}  & \textbf{17.02}    \\ 
\midrule
\multirow{9}{*}{Decompress} & Coder Decode & 5932 & 0.52 & RAM to GPU      & 11101  & 0.28  \\
 & Latent to GPU  & 96254  & 0.03   & Coder Decode  & 11091 & 0.28  \\
 & VQ-VAE Decode  & 720   & 4.27  & VQ-VAE Decode  & 721  & 4.26   \\
& Distribution to RAM & 2426   & 1.27   & Coder Decode   & 672  & 4.57   \\
& Coder Decode    & 624   & 4.93   & AR Decode  & 869 & 3.53    \\
& Residual to GPU  & 8867 & 0.35   & GPU to RAM & 2521 & 1.20 \\
& AR Decode  & 849  & 3.62  & -   & -   & -    \\
& GPU to RAM    & 2515  & 1.22  & -    & -   & -   \\
 & \textbf{Total} & \textbf{186} & \textbf{16.21} & \textbf{Total}  & \textbf{217}  & \textbf{14.12}   \\
\bottomrule
\end{tabular}
\end{table*}

\begin{table}[t]
\small
\centering
\caption{The effect of different codebook size, dimension and model capacity. Theoretical BPD and throughput
on CIFAR10 are reported.}
\label{tab:resblock_and_codebook}
\begin{tabular}{cccclc} 
\toprule
\multirow{2}{*}{\begin{tabular}[c]{@{}c@{}}Mid Channel \\Dim\end{tabular}} & \multirow{2}{*}{Blocks} & \multicolumn{2}{c}{Codebook} & \multicolumn{1}{c}{\multirow{2}{*}{BPD}} & \multirow{2}{*}{\begin{tabular}[c]{@{}c@{}}Throughput\\~(MB/s)\end{tabular}}  \\   & & Size & Dim & \multicolumn{1}{c}{}    &   \\ 
\midrule
\textbf{32} & \textbf{4}  & \textbf{256} & \textbf{32}  & \textbf{4.17} & \textbf{277}   \\
32 & 4  & 128 & 32  & 4.19 & 306  \\
32 & 4  & 256 & 16  & 4.18 & 293  \\
32 & 4  & 256 & 64  &  4.18  &  249  \\
32 & 4  & 512 & 32  &  4.19   &  232 \\
32 & 4  & 512 & 64  &  4.19  &   210  \\
\midrule
32 & 8  & 256 & 32  &  4.16   &  177 \\
64 & 4  & 256 & 32  &  4.14   &  132  \\
64 & 8  & 256 & 32  &  4.10   &  74  \\
64 & 16  & 256 & 32  &  4.06   &  41  \\
\bottomrule
\end{tabular}
\end{table}

{\em Decompress.} The CPU Coder requires two additional RAM-GPU transfers during decompression because the Coder decoding is done on the CPU. The CPU Coder decodes latent vector indexes before sending them to the GPU. Then we need to transmit the distribution parameters to RAM after VQ-VAE decodes them. The residual is then decoded by the Coder and passed back to the GPU. The Three-Way Auto-regressive recovers original data from residual data. The GPU Coder, on the other hand, only requires two RAM-GPU transfers, and all Coder decoding is done on the GPU, which saves a lot of time.

\begin{figure}[t]
    \centering 
    \includegraphics[width=0.40\textwidth]{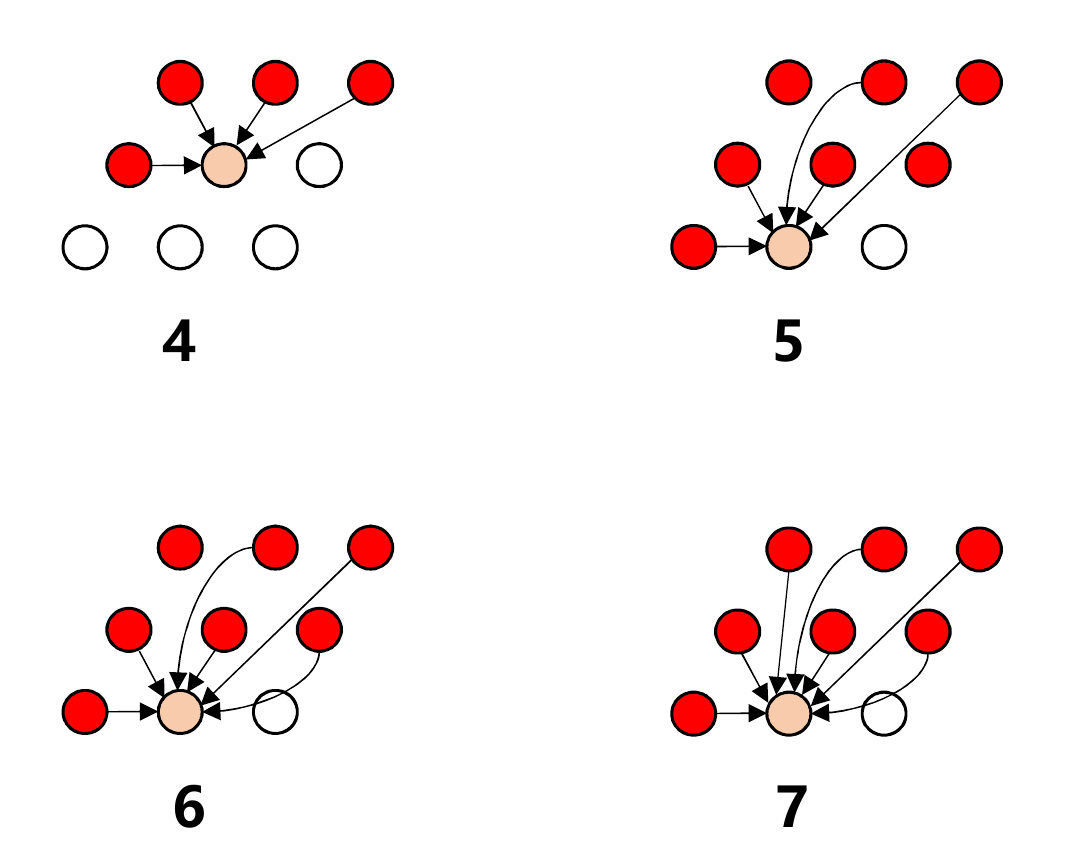}
    \caption{The receptive field rules of the auto-regressive model. Receptive field from 4 to 7. }
    \label{fig:receptive_field}
\end{figure}

\begin{table}[t]
\centering
\caption{Ablation study on the receptive field of AR model and the parallel mechanism. Theoretical BPD and decompress throughput on CIFAR10 (red channel) are reported. Receptive field means the current point is predicted by how many previous points.}
\label{tab:w/o_parallel}
\begin{tabular}{ccr} 
\toprule
\multicolumn{1}{c}{Receptive Field} & BPD    & \multicolumn{1}{c}{\begin{tabular}[c]{@{}c@{}}Throughput\\(MB/s)\end{tabular}}  \\ 
\midrule
3 (with parallel)                  & 5.7714 & \textbf{382.5}                                                                  \\
3                & 5.7714 &       48.5                                                                 \\
4                                  & 5.7737 & 47.5                                                                            \\
5                                  & 5.7708 & 46.0                                                                            \\
6                                  & 5.7701 & 44.8                                                                            \\
7                                  & 5.7449 & 44.0                                                                            \\
\bottomrule
\end{tabular}
\end{table}

\begin{table}
\centering
\caption{Different settings of how to predict residual distribution. \ie, different inputs of VQ-VAE. We use original image, reconstruction image, residual, and original image concatenated with residual to predict the residual. }
\begin{tabular}{ll} 
\toprule
Input                         & BPD   \\ 
\midrule
Original Image                & \textbf{4.17}  \\
Reconstruction Image          & 4.23  \\
Residual                      & 4.19  \\
Origanl Image concat Residual & 4.17  \\
\bottomrule
\end{tabular}
\end{table}

\section{Model architecture}
We do further experiments to see how different codebook sizes and dimensions affect the results.
A large codebook size or codebook dimension favors the BPD but reduces performance, as seen in \cref{tab:resblock_and_codebook}. A tiny codebook size or codebook dimension does indeed produce substantially faster throughput, but it would contains relatively little information when applying our model to high resolution images. As a result, on codebook size/dimension, there is a trade-off between BPD and throughput. To balance the trade-off, we set the codebook size to 256 and the codebook dimension to 32 in our framework.

We investigate the impact of various model capacities. We modify the number of residual block and the mid channel dimension to change the model capacity. As shown in \cref{tab:resblock_and_codebook}, the larger model achieves lower BPD, but the throughput of the model drops even faster, which deviates from the 'practical' aim. As a result, we adjust our model to 4 residual blocks and 32 mid channel dimensions.

\section{Receptive field rules of AR model}
\Cref{fig:receptive_field} demonstrates the receptive field rules of the auto-regressive model on a single channel which we use in the ablation study. As the figure suggests, these receptive field settings cannot be parallelized using what we implement in Three-Way Auto-regressive. 

\section{Parallel vs No Parallel}
We add the experiment result of Three-Way Auto-regressive without parallel. As shown in \cref{tab:w/o_parallel}, Three-Way Auto-regressive using our designed parallel mechanism achieves the fastest throughput with competitive compression ratio.

\section{Why predict residual given original image?}
We investigate the different input of VQ-VAE model. We use the original image, reconstruction image, residual, and original image concatenated with residual to predict the residual. As illustrated in \cref{tab:vqvae_diff_inputs}, directly using residual to predict residual is not the best solution. As the original image has a lot of spatial information, we choose to model the residual given original image.
\end{document}